\documentclass{article}
\usepackage{amsmath}
\usepackage{amsfonts}
\usepackage{amssymb}
\usepackage{amsthm}
\usepackage{latexsym}
\usepackage{graphicx}
\usepackage[a4paper]{geometry}
\usepackage[utf8]{inputenc}

\newtheorem{thm}{Theorem}
\newtheorem{claim}{Claim}[thm]
\newtheorem{lem}[thm]{Lemma}
\newtheorem{proposition}[thm]{Proposition}

\newcommand{\RE}{\mathbb{R}}			

\newcommand{\GG}{\mathcal{G}}

\newcommand{\BB}{\mathcal{B}}
\newcommand{\NN}{\mathcal{N}}

\title{The Cost of Perfection for Matchings in Graphs}
\author{
Emilio Vital Brazil\thanks{Department of Computer Science, University of Calgary, Canada.}
\and
Guilherme D. da Fonseca\thanks{LIRMM, Universit\'e Montpellier 2, France.}
\and
Celina M. H. de Figueiredo\thanks{COPPE, Univ.~Federal do Rio de Janeiro, Brazil.}
\and
Diana Sasaki\footnotemark[3]
}

\begin{document}

\maketitle

\begin{abstract}
Perfect matchings and maximum weight matchings are two fundamental combinatorial structures. 
We consider the ratio between the maximum weight of a perfect matching and the maximum weight of a general matching. 
Motivated by the computer graphics application in triangle meshes, where we seek to convert a triangulation into a quadrangulation by merging pairs of adjacent triangles, we focus mainly on bridgeless cubic graphs.
	
First, we characterize graphs that attain the extreme ratios. 
Second, we present a lower bound for all bridgeless cubic graphs.
Third, we present upper bounds for subclasses of bridgeless cubic graphs, most of which are shown to be tight.
Additionally, we present tight bounds for the class of regular bipartite graphs.
\end{abstract}

\section{Introduction}
\label{Introduction}
The study of matchings in cubic graphs has a long history in combinatorics, dating back to Petersen's theorem~\cite{konig}. 
Recently, the problem has found several applications in computer graphics and geographic information systems~\cite{biedl01,gopi04,remacle11,daniels11}.
Before presenting the contributions of this paper, we consider the following motivating example in the area of computer graphics.

Triangle meshes are often used to model solid objects. 
Nevertheless, quadrangulations are more appropriate than triangulations for some applications~\cite{daniels11, tri2quad}. 
In such situations, we can convert a triangulation into a quadrangulation by merging pairs of adjacent triangles (Figure~\ref{fig:bunny}). 
Hence, the problem can be modeled as a matching problem by considering the dual graph of the triangulation, where each triangle corresponds to a vertex and edges exist between adjacent triangles. 
The dual graph of a triangle mesh is a bridgeless cubic graph, for which Petersen's theorem guarantees that a perfect matching always exists~\cite{biedl01,bm}. 
Also, such a matching can be computed in $O(n \log^2 n)$ time~\cite{diks10}.

\begin{figure}[t]
\centering
\includegraphics[width = .7\textwidth]{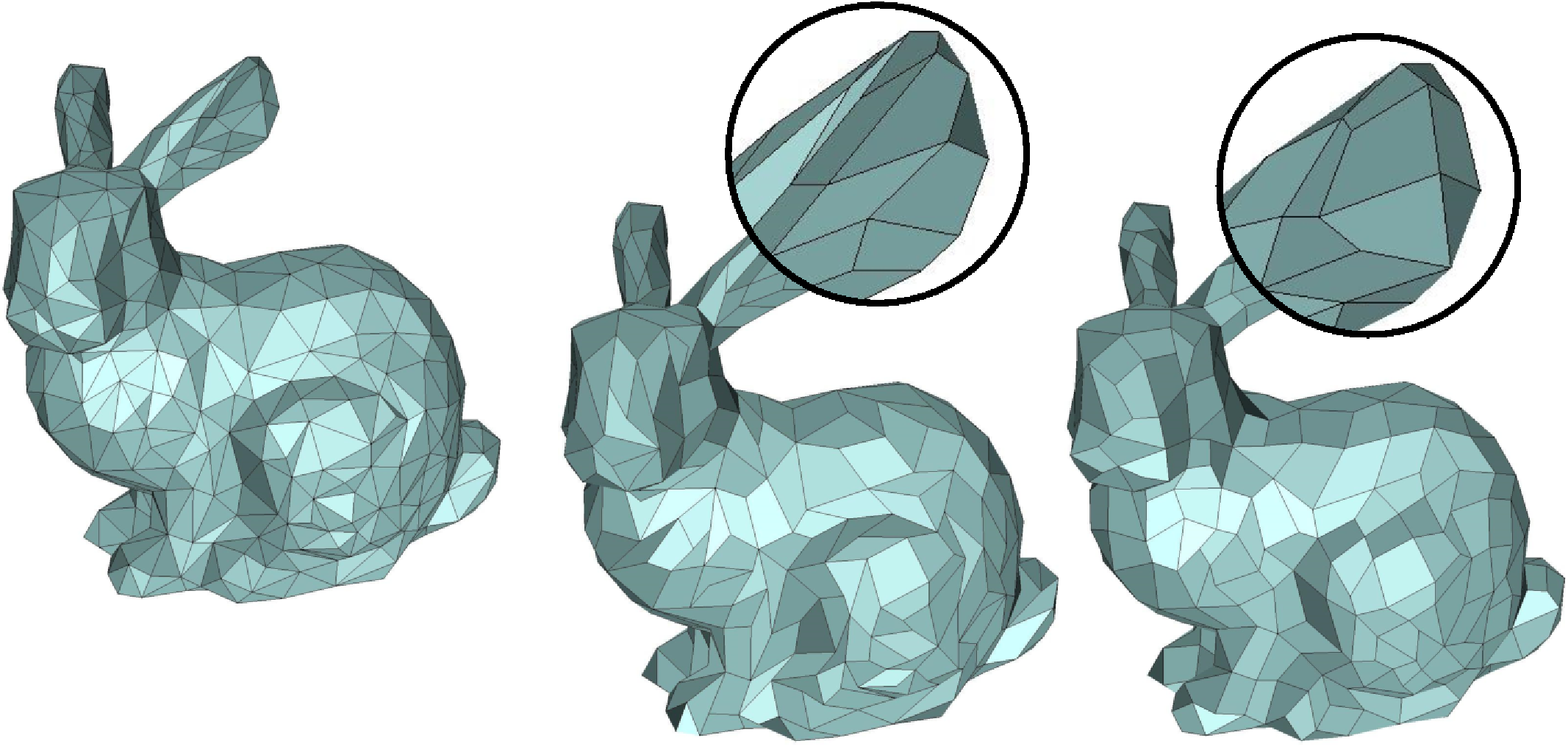}
\caption{Stanford Bunny Model: triangular mesh (left) and two quadrangular meshes.}
\label{fig:bunny}
\end{figure}

Unfortunately, from the computer graphics perspective, some pairs of triangles lead to undesirable quadrilaterals (for example, when the triangles are skinny or lie on very different planes). 
A natural extension to the cubic graph model assigns a weight to each edge (i.e., to each pair of adjacent triangles), which expresses how desirable the corresponding quadrilateral is. 
In Figure~\ref{fig:bunny} (middle and right) we can compare the results when two different weight functions are used to create  quadrangular meshes, observe that the middle one has more skinny quadrilaterals than the right one.
However, even when using good weight functions, an inherent difficulty arises: The maximum weight matching may not be a perfect matching.

In this paper, we study the relationship between these two types of matchings, in order to understand how much worse (in terms of total weight) we do by selecting the maximum weight perfect matching instead of the maximum weight matching. The interest of such study goes beyond the original computer graphics application, raising intriguing theoretical questions.

We provide bounds for the ratio between the maximum weight of a perfect matching and the maximum weight of a matching. We take advantage of the existing rich literature about bridgeless cubic graphs, a historical graph class much studied in the context of important graph theory conjectures, such as: The Four Color Conjecture~\cite{Appel}, the Berge-Fulkerson Conjecture, and the Cycle Double Cover Conjecture~\cite{Celmins}. We formalize the aforementioned concepts in the next paragraphs, after some definitions.

Let $G=(V,E)$ be a connected undirected graph. A \emph{bridge} is an edge $uv \in E$ such that all paths between $u$ and $v$ go through $uv$. A graph is \emph{bridgeless} if it has no bridges. A graph is \emph{$\Delta$-regular} if every vertex has degree exactly $\Delta$. A $3$-regular graph is called a \emph{cubic} graph. A cubic graph is bridgeless if and only if it is biconnected~\cite{bm}.

A \emph{matching} in $G$ is a set $M \subseteq E$ such that no two edges in $M$ share a common vertex. Recall that given a matching $M$ in a graph $G$, we say that $M$ \emph{saturates} a vertex $v$ and that vertex $v$ is \emph{$M$-saturated}, if some edge of $M$ is incident to $v$~\cite{bm}.
A matching $P$ is \emph{perfect} if $|P| = |V|/2$. 
A matching is \emph{maximal} if it is not a subset of any other matching and is \emph{maximum} if it has maximum cardinality. 
 A cubic graph $G$ is \emph{Tait-colorable} if the edges of $G$ can be
partitioned into three perfect matchings, all Tait-colorable graphs
are bridgeless~\cite{bm}. A \emph{snark} is a bridgeless cubic graph
that is not Tait-colorable and the smallest snark is the Petersen graph~\cite{isaacs75}.

Let $w : E \rightarrow \RE^+$ be the \emph{weight} of the edges. It will be convenient to allow for the weight of some edges to be zero as long as there is at least one edge with nonzero weight. Given a subset $E' \subseteq E$, we refer to the quantity $w(E') = \sum_{e \in E'} w(e)$ as the \emph{weight} of $E'$. A \emph{maximum weight matching} is a matching $M^*(G)$ of maximum possible weight in $G$. A \emph{maximum weight perfect matching} is a perfect matching $P^*(G)$ of maximum possible weight (among all perfect matchings of $G$). Given a graph $G$ which admits a perfect matching, we define
$$\eta(G) =  \min_{w:E\rightarrow \RE^+} \frac{w(P^*(G))}{w(M^*(G))}.$$

The value of $\eta(G)$ can be as small as $0$. To see that, consider the path of length $3$ where the middle edge has weight $1$ and the two remaining edges have weight $0$.  The graph $G$ has a single perfect matching $P$ with weight $w(P)=0$, while there is a non-perfect matching with weight $1$. Note that we allow edge weights to be $0$, for otherwise, $\eta(G)$ could be made arbitrarily small as the weights approach $0$, and the minimum would never be attained. By allowing edge weights to be $0$, we show that the minimum is always attained (Theorem~\ref{thm_min}).

A graph $G$ with $\eta(G)=0$ represents one extreme of the problem. In this case, requiring a matching to be perfect may result in a matching with zero weight, where a matching with arbitrarily high weight may exist. In the other extreme, we have graphs $G$ with $\eta(G) = 1$. In this case, for every $w$ there is a perfect matching with the same weight as the maximum weight matching. In Section~\ref{extreme}, we give precise characterizations of these two extremes. In the remainder of the paper, we manage to determine the exact value of $\eta$ for several graphs that lie in between the two extreme cases. Some examples are presented in Figure~\ref{fig:examples}.

\begin{figure}[t]
\centering
\includegraphics[width = .9\textwidth]{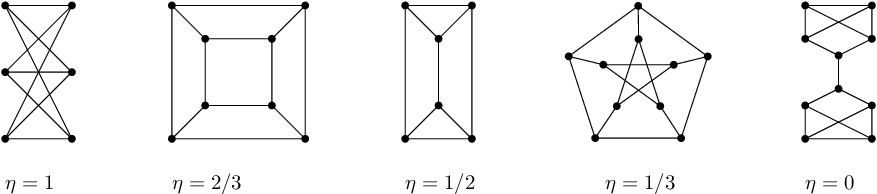}
\caption{Exact value of $\eta$ for five small cubic graphs.}
\label{fig:examples}
\end{figure}

Consider a graph $G$ that is known to be a member of a graph class $\GG$. Since $\eta(G)$ is only defined for graphs that admit a perfect matching, we assume that all graphs in $\GG$ admit perfect matchings. Different graphs $G,G' \in \GG$ may have $\eta(G) \neq \eta(G')$. We define the value of $\eta(\GG)$ for a graph class $\GG$ as:
$$\eta(\GG) = \inf_{G \in \GG} \eta(G).$$
Sometimes, when the graph $G$ or the graph class $\GG$ is clear from the context, we refer to $\eta(G)$ or $\eta(\GG)$ simply as $\eta$.

The infimum in the previous definition is not always attained by a graph in the class. For instance, a consequence of Theorem~\ref{lem:lbip} is that the value of $\eta$ for the class of all regular bipartite graphs is $0$, but by Theorem~\ref{thm:eta0} no graph in that class has $\eta = 0$.

The infimum value bounds the worst case scenario, in terms of our motivation of approximating a maximum weight matching with a perfect matching. In many applications it is not possible to know a priori all graphs that will be given, thus the worst case scenario provides very useful information.  We also note that the infimum value is theoretically richer than its counterpart, the supremum value of $\eta(G)$ for $G \in \GG$. The supremum value is often equal to $1$, since most natural graph classes contain one of the following small graphs which have $\eta = 1$: $K_2,K_4,C_4, K_{3,3}$.

An immediate consequence of this definition is that given two graph classes $\GG_1 \subseteq \GG_2$, we have $\eta(\GG_1) \geq \eta(\GG_2)$. 
Thus, in order to prove bounds on $\eta$ that apply to as many graph classes as possible, it is useful to obtain lower bounds on $\eta$ for ``large'' graph classes and upper bounds on $\eta$ for ``small'' graph classes. 
Therefore, we will treat the upper bounds as valuable knowledge about the whole graph class even when it is proved only for a single example, although most of the obtained upper bounds were extended to infinite families of cubic graphs.

Section~\ref{cubic} is devoted to bounds for the target class of bridgeless
cubic graphs. We show that $\eta(\GG) \geq 1/3$ where $\GG$ is the class of bridgeless cubic graphs, and therefore to all its subclasses.
We show that $\eta(G) = 1/3$ for a particular planar hamiltonian cubic graph $G$, and hence for all the classes that contain it. 
Note that both planar  bridgeless cubic graphs and hamiltonian cubic graphs are Tait-colorable~\cite{bm}. 
We also show that the Petersen graph has $\eta = 1/3$, and hence the class of generalized Petersen graphs has $\eta = 1/3$. Notice that the Petersen graph is the only graph in the class of generalized Petersen graphs that is not Tait-colorable~\cite{Castagna1972}. Since the class of bridgeless cubic graphs is divided into Tait-colorable graphs and snarks, we investigate two well-known families of snarks, Blanu\v{s}a First and Blanu\v{s}a Second, and present an infinite family of snarks obtained by the 2-construction for which all members satisfy $\eta = 1/3$. Furthermore, we show that $\eta(\NN) = 1/3$, where $\NN$ is the class of nonhamiltonian bipartite cubic graphs.
In Section~\ref{s:regular}, we show that for all $\Delta \geq 2$, the class $\BB_\Delta$ of $\Delta$-regular bipartite graphs satisfies $\eta(\BB_\Delta) = 1/\Delta$. 

\section{Extreme Cases}
\label{extreme}

Given a graph $G$ that admits a perfect matching, we recall that the parameter $\eta$ is defined as
$$\eta(G) =  \min_{w:E\rightarrow \RE^+} \frac{w(P^*(G))}{w(M^*(G))},$$
where $P^*(G)$ is the maximum weight perfect matching, $M^*(G)$ is the maximum weight matching, and $w$ is the non-negative weight function, with at least one edge of non-zero weight. We start this section by showing that $\eta(G)$ is well defined for all graphs $G$ that admit a perfect matching.

\begin{thm}
\label{thm_min}
The \emph{minimum} in the definition of $\eta(G)$ is attained for all
graphs $G$ that admit a perfect matching.
\end{thm}
\begin{proof}
Fix a graph $G = (V,E)$.
For ease of notation, let $W := \{ w:E\to\mathbb{R}^+ \;;\;$ there is $e \in E$ such that $w(e)>0\}$, $\mathcal{M} := \{M \subseteq E \;;\; M$ is a matching$\}$, and $\mathcal{P} := \{P \in \mathcal{M} \;;\; P$ is a perfect matching\}.
Note that $\mathcal{M}$ and $\mathcal{P}$ are finite.

Let $Q := \{ w(P^*)/w(M^*) \;;\; w\in W$, $P^*$ is a maximum weight perfect matching for $w$, and $M^*$ is a maximum weight matching for $w\}$, and let $Q' := \{ w(P^*) \;;\; w\in W$, $P^*$ is a maximum weight perfect matching for $w$, and there is a matching $M$ such that $w(M) = 1$ and $ w(E \smallsetminus M)=0 \}$.
	
\begin{claim}
	$\inf Q = \inf Q'$.
\end{claim}
Indeed, since $Q' \subseteq Q$ we immediately get $\inf Q \leq \inf Q'$.
For the converse, we will show that for any $q \in Q$ there is $q' \in Q'$ such that $q' \leq q$.
Let $q = w(P^*)/w(M^*) \in Q$, and define a new weight function $w'$ by 
\[w'(e) = 
	\begin{cases}
		0,& \text{if } e \not \in M^* \\
		w(e)/w(M^*), &\text{otherwise}.
	\end{cases} \]
Then $M^*$ is still a maximum weight matching for $w'$, and furthermore $w'(M^*) = 1$ and $w'(E \smallsetminus M^*) = 0$.
Finally, if $P'$ is a maximum weight perfect matching for $w'$ then $w'(P') \in Q'$, and by construction we have 
\[\begin{array}{rcl}
		w'(P') 
			&\leq& w(P')/w(M^*) \\
			&\leq& w(P^*)/w(M^*) \\
			&=&q,
	\end{array} \]
which concludes the proof of the claim.

Given a matching $M$ and a perfect matching $P$ of $G$, let $Q'_{M,P} := \{ w(P) \;;\; w\in W$, $P$ is a maximum weight perfect matching for $w$, $w(M)=1$, and $w(E \smallsetminus M)=0\}$.
Note that	
	\[ Q' = \displaystyle\bigcup_{M \in \mathcal{M}} \bigcup_{P \in \mathcal{P}} Q'_{M,P}, \]
and therefore we have
\[	\begin{array}{rcl}
			\inf Q
				&=& \inf Q' \\
				&=& \displaystyle\inf_{M \in \mathcal{M}} \inf_{P \in \mathcal{P}} \inf Q'_{M,P} \\
				&=& \displaystyle\min_{M \in \mathcal{M}} \min_{P \in \mathcal{P}} \inf Q'_{M,P}.
		\end{array} \]
Thus, in order to conclude that $Q$ has a minimum it is enough to conclude that each $Q'_{M,P}$ has a minimum, which we will do by showing that each $Q'_{M,P}$ is compact.

Given $M = \{e_1,\ldots,e_k\}$ and $P \in \mathcal{P}$, consider the set $K := \{ (w(e_1),\ldots,w(e_k)) \in \mathbb{R}^k \;;\; w\in W$, $P$ is a maximum weight perfect matching for $w$, $w(M)=1$, and $w(E \smallsetminus M)=0\}$.

\begin{claim}
	$K$ is compact.
\end{claim}
Indeed, since $K \subseteq [0,1]^k$, it is enough to show that $K$ is closed.
To this end, let $(w_n(e_1),\ldots,w_n(e_k))_{n\in \mathbb{N}}$ be a sequence of elements of $K$ which converges to $x = (x_1,\ldots,x_k) \in \mathbb{R}^k$, i.e., such that $x_i = \lim_{n \in \mathbb{N}} w_n(e_i)$ for each $i = 1,\ldots,k$.
Define a weight function $w$ by 
	\[ 	w(e) = 
			\begin{cases}
				0,&\text{if } e \not \in M \\
				x_i,&\text{if } e = e_i.
			\end{cases} \]
To conclude that $x \in K$, we have to show (i) $w(E \smallsetminus M) = 0$, (ii) $w(M) = 1$, (iii) $w \in W$, and (iv) $P$ is a maximum weight perfect matching for $w$.
By construction we have (i). 
For (ii), note that 
\[\begin{array}{rcl}
		w(M) 
		&=& \displaystyle\sum_{e\in M} w(e) \\
		&=& \displaystyle\sum_{e\in M} \lim_{n \in \mathbb{N}} w_n(e) \\
		&=& \displaystyle\lim_{n \in \mathbb{N}}\sum_{e\in M}  w_n(e) \\
		&=& \displaystyle\lim_{n \in \mathbb{N}}w_n(M) \\
		&=& 1,
	\end{array} \]
from which we can also directly conclude (iii).
Finally, for (iv), note that for any perfect matching $P'$, we have 
\[\begin{array}{rcl}
		w(P') 
		&=& \displaystyle\sum_{e\in P'} w(e) \\
		&=& \displaystyle\sum_{e\in P'} \lim_{n \in \mathbb{N}} w_n(e) \\
		&=& \displaystyle\lim_{n \in \mathbb{N}}\sum_{e\in P'}  w_n(e) \\
		&=& \displaystyle\lim_{n \in \mathbb{N}}w_n(P'). \\
	\end{array} \]
Therefore, since for all $n$ we have $w_n(P) \geq w_n(P')$, it follows that 
\[\begin{array}{rcl}
		w(P) 
		&=& \displaystyle\lim_{n \in \mathbb{N}}w_n(P) \\
		&\geq& \displaystyle\lim_{n \in \mathbb{N}}w_n(P') \\
		&=& w(P')
	\end{array} \]
which concludes the proof that $K$ is compact.

But now, since 
	\[ Q'_{M,P} = \left\{\sum_{\substack{1 \leq i \leq k \\e_i \in P}} x_i \;;\; (x_1,\ldots,x_k) \in K \right\},\] 
it follows that $Q'_{M,P}$ is compact, since it is the continuous image of a compact set.
This concludes the proof.
\end{proof}

We defined $\eta(G)$ for a graph $G$ which has a perfect matching. By definition $0 \leq \eta(G) \leq 1$. In order to get used to the definition of $\eta$, we characterize the graphs that attain the extreme values of $\eta$. First, we characterize the graphs $G$ with $\eta(G) = 1$.

\begin{thm}
\label{thm:eta1}
A graph $G$ has $\eta(G) = 1$ if and only if all connected components of $G$ are $K_{2n}$ or $K_{n,n}$ for $n \geq 1$.
\end{thm}
\begin{proof}
Sumner~\cite{sumner79} showed that a graph $G$ satisfies that every maximal matching is a perfect matching if and only if every connected component of $G$ is a $K_{2n}$ or $K_{n,n}$ for $n \geq 1$. Therefore, it suffices to show that $\eta(G)=1$ if and only if every maximal matching of $G$ is a perfect matching.

First, we prove that if $\eta(G) = 1$, then every maximal matching of $G$ is a perfect matching. For the sake of a contradiction, suppose $\eta(G) = 1$ and let $M$ be a maximal matching that is not a perfect matching. Define $w(e) = 1$ if $e \in M$ and $w(e) = 0$ otherwise. For any perfect matching $P$ of $G$ it holds that $M \not\subseteq P$. Thus, there is at least one edge $e \in M$ such that $e \not \in P$. Therefore, $w(M) > w(P)$ and consequently $\eta(G) < 1$.

Second, we prove that if every maximal matching of $G$ is a perfect matching, then $\eta(G) = 1$. For a fixed weight function $w : E \rightarrow \RE^+$, let $M^*(G)$ be the matching of maximum weight. If there are edges with zero weight, it is possible that $M^*(G)$ is not a perfect matching. Nevertheless, there is a perfect matching of maximum weight $P^*(G) \supseteq M^*(G)$ which can be obtained from $M^*(G)$ by including edges of zero weight. Consequently, $w(P^*(G)) = w(M^*(G))$ and $\eta(G)~=~1$.
\end{proof}

Note that, if we allow only positive nonzero weights, then every matching of maximum weight is a maximal matching.
The condition in the proof of Theorem~\ref{thm:eta1} that every maximal matching is a perfect matching now implies that every matching of maximum weight is actually perfect, the sets of matchings of type $M^*$ and $P^*$ are equal, and sufficiency is immediate.
However, if we allow negative weights, then sufficiency does not hold.
On the other hand, necessity holds regardless of allowing zero or negative weights.

Next, we characterize graphs $G$ with $\eta(G) = 0$.

\begin{thm} \label{thm:eta0}
A graph $G$ has $\eta(G) = 0$ if and only if there is an edge $e \in E$ that is not contained in any perfect matching.
\end{thm}
\begin{proof}
First, we prove that if for every edge $e \in E$ there is a perfect matching that contains $e$, then $\eta(G) > 0$. We remind the reader that by definition of the weight function, at least one edge $e$ has $w(e) > 0$. Therefore, there is a perfect matching $P$ that contains $e$ and have $w(P) > 0$. Consequently, $\eta(G) > 0$.

Second, we prove that if there is an edge $e \in E$ that is not in any perfect matching, then $\eta(G) = 0$. Let the weight of $e$ be $1$ and the weight of all other edges be $0$. In this case, all perfect matchings have weight $0$ and the maximum weight matching has weight $1$. Consequently, $\eta(G) = 0$.
\end{proof}

The previous theorem implies, for example, that $\eta(G) = 0$ for every cubic graph $G$ which contains a bridge and admits a perfect matching, since an edge that is adjacent to the bridge is not in any perfect matching. Cubic graphs with all bridges on a single path admit a perfect matching~\cite{errera22}.

\section{Bridgeless Cubic Graphs}
\label{cubic}

In this section, we provide upper and lower bounds on $\eta$ for our target class of bridgeless cubic graphs. We start with a lower bound for arbitrary bridgeless cubic graphs. Remark that the lower bound extends to all subclasses.

Clearly, if $G$ is Tait-colorable, then each edge is contained in a perfect matching, which implies that $\eta(G)>0$. Actually, we get a better lower bound, since a Tait-colorable graph $G$ admits 3 perfect matchings so that each edge is covered precisely once, which gives $\eta(G) \geq 1/3$. The famous Berge-Fulkerson Conjecture~\cite{Fulkerson1971,Giuseppe} says that every bridgeless cubic graph admits a family of 6 perfect matchings such that each edge is covered precisely twice.

The proof of Lemma~\ref{lem:lb} establishes the lower bound of $\eta(G) \geq 1/3$, for an arbitrary bridgeless cubic graph, by using a property~\cite{edmonds65} more general than a Tait-coloring but weaker than the Berge-Fulkerson Conjecture.

\begin{lem} \label{lem:lb}
Let $G$ be a bridgeless cubic graph. Then, $\eta(G) \geq 1/3$.
\end{lem}
\begin{proof}
It is known that given a bridgeless cubic graph $G$, there is an integer $k$ (depending on $G$) such that $G$ has a family of $3k$ perfect matchings that cover each edge of $G$ exactly $k$ times~\cite{edmonds65}. Let $P_1,\ldots,P_{3k}$ denote such perfect matchings. Assume without loss of generality that $w(P_1) \geq \cdots \geq w(P_{3k})$. Let $M^*(G)$ be the maximum weight matching. We have that
\[w(M^*(G)) \leq \frac{w(P_1) + \cdots + w(P_{3k})}{k} \leq 3\;w(P_1) \leq 3\;w(P^{*}(G))\]
and therefore $\eta(G) \geq 1/3$.
\end{proof}

Since upper bounds on $\eta$ extend to superclasses, it is useful to prove upper bounds for graphs that are contained in several relevant classes. A particular subclass of bridgeless cubic graphs is the class of Tait-colorable graphs. Two subclasses of Tait-colorable graphs are planar bridgeless cubic graphs and hamiltonian cubic graphs. We start by proving a tight bound for the intersection of the two aforementioned classes.

\begin{lem} \label{lem:ub}
There are infinitely many planar hamiltonian cubic graphs $G$ with $\eta(G) = 1/3$.
\end{lem}
\begin{proof}
First, let $G$ be the cubic graph represented in Figure~\ref{fig:tait13}(a). 
Note that $G$ is planar and hamiltonian (see Figure~\ref{fig:tait13}(b)). 
By Lemma~\ref{lem:lb}, $\eta(G) \geq 1/3$. We now show that $\eta(G) \leq 1/3$. Let $e_1,e_2,e_3$ and $v_{1,2},v_{2,3},v_{1,3}$ be the edges and vertices labeled in Figure~\ref{fig:tait13}(c). 
We can set $w(e_1)=w(e_2)=w(e_3)=1$ and set all other edge weights to $0$. 
A perfect matching may contain at most one of $e_1,e_2,e_3$. 
To see that, note that if a matching contains two such edges $e_i,e_j$, then vertex $v_{i,j}$ indicated in Figure~\ref{fig:tait13}(c) cannot be saturated by any edge of the matching. 
Therefore, there is a matching $e_1,e_2,e_3$ of weight $3$ and a perfect matching may have weight at most~$1$, which implies that $\eta(G) \leq 1/3$.

\begin{figure}[ht]
\centering
\includegraphics[scale = .35]{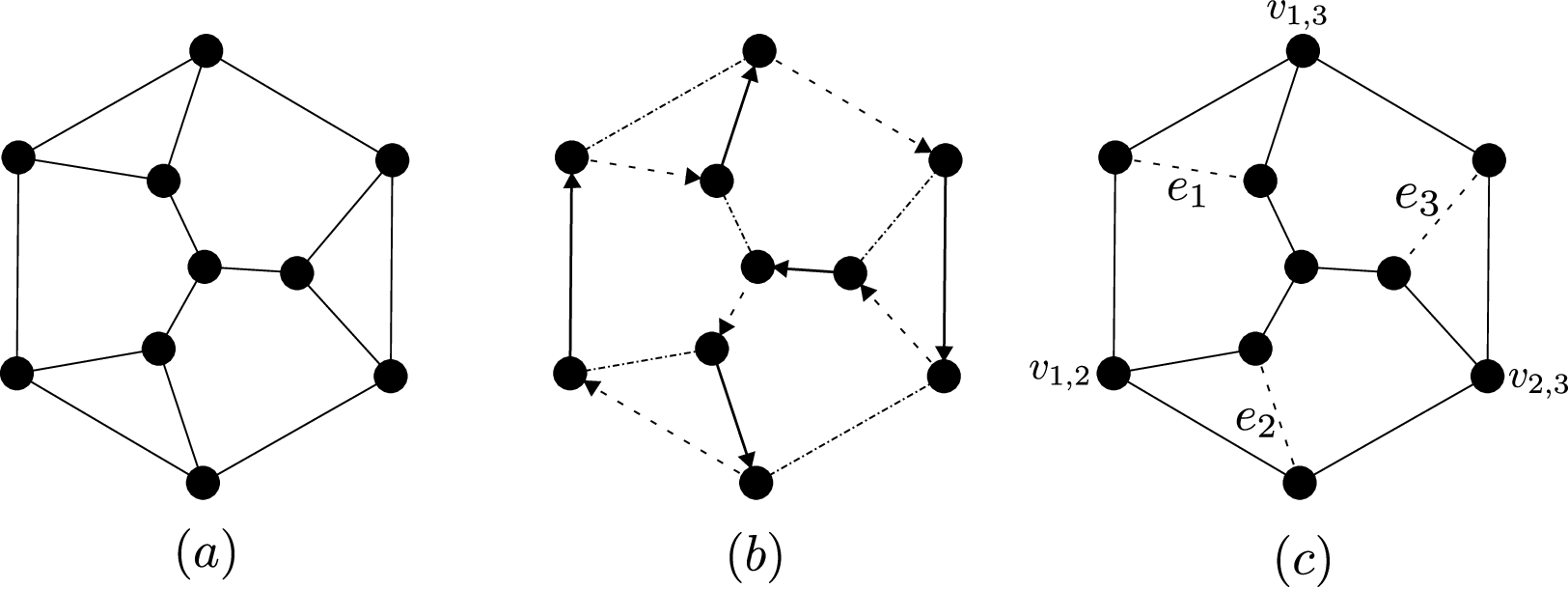}
\caption{(a)~Tait-colorable cubic graph $G$ used in the proof of Lemma~\ref{lem:ub}. (b)~Tait-coloring and hamiltonian cycle of $G$. (c)~Edges and vertices used in the proof of Lemma~\ref{lem:ub}.}
\label{fig:tait13}
\end{figure}

One way to obtain infinitely many such graphs is to remove the central vertex of the graph in Figure~\ref{fig:tait13} and connect the remaining graph through a matching of size 3 to another planar hamiltonian cubic graph with one vertex removed. Another way to obtain infinitely many such graphs is to remove an edge incident to the central vertex of this graph and connect the remaining graph through a matching of size 2 to another planar hamiltonian cubic graph with one edge removed. Both constructions are classical in the theory of cubic graphs~\cite{isaacs75}.
\end{proof}

The upper bound of the next result for bipartite cubic graphs uses an easy but powerful counting argument that will be generalized next. While the argument generally gives rise to high upper bounds, in some cases, as the cube graph, the bound is in fact tight.

\begin{proposition} \label{prop:cube}
Let $Q$ be the cube graph. Then, $\eta(Q) = 2/3$.
\end{proposition}
\begin{proof}
First, we show that $\eta(Q) \leq 2/3$. Let $e_1,e_2,e_3$ be the edges labeled in Figure~\ref{fig:cube}. We can set $w(e_1)=w(e_2)=w(e_3)=1$ and set all other edge weights to $0$. 
Since $Q\smallsetminus \{e_1,e_2,e_3\}$ is an independent set, any perfect matching of $Q$ may contain at most two edges among $e_1,e_2,e_3$.

\begin{figure}[ht]
\centering
\includegraphics[width = .15\linewidth]{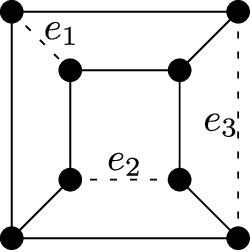}
\caption{Cube graph $Q$ with edges used in the proof of Proposition~\ref{prop:cube} marked.
}
\label{fig:cube}
\end{figure}

Second, we show that $\eta(Q) \geq 2/3$. We consider the weight assignment $w$ that defines $\eta$. Let $M^*$ and $P^*$ respectively be a maximum weight matching and a maximum weight perfect matching with respect to $w$. We assume without loss of generality that $M^*$ is maximal. Since $M^*$ is maximal, it contains at least $3$ edges. If $M^*$ is perfect, than $\eta(Q) = w(P^*)/w(M^*) = 1$. Therefore, $M^*$ must contain exactly $3$ edges. By removing from $M^*$ the edge of minimum weight, we obtain a matching $M'$ with $w(M') \geq 2 w(M^*)/3$. Since $M'$ has two edges, it can be extended to a perfect matching $P' \supset M'$ with $w(P') \geq 2 w(M^*)/3$, proving that $\eta(Q) \geq 2/3$.
\end{proof}

Next, we introduce a lemma that generalizes the upper bound from Proposition~\ref{prop:cube} and will help us to prove upper bounds for two cubic graphs in the well-known class of generalized Petersen graphs~\cite{Watkins_gen}.

\begin{lem}\label{maxmatchings}
Let $M$ be a maximal matching of a bridgeless cubic graph $G$ and $S=V \smallsetminus \{M$-saturated vertices$\}$ be the corresponding independent set. We have the upper bound:

$$\eta(G)\leq \frac{\left|V\right|-2|S|}{|V|-|S|}$$
\end{lem}
\begin{proof}
Since $S$ is an independent set, any perfect matching must have exactly $|S|$ edges not of $M$, each of which with exactly one end vertex in $S$, and therefore at most $\frac{|V|}{2}-|S|$ edges of any perfect matching are in $M$. Every cubic graph has an even number of vertices which implies that $|S|$ is even and $|M|=\frac{|V|}{2}-\frac{|S|}{2}$ and the lemma follows by setting $w(e)=1$, if $e \in M$, $w(e)=0$, otherwise.
\end{proof}

A generalized Petersen graph $G(n,k)$, for $n \geq 3$ and $1 \leq k \leq \lfloor (n-1)/2 \rfloor$, is a graph with vertex set $\{u_0,\ldots,u_{n-1},v_1,\ldots,v_{n-1}\}$, and the following three types of edges for $0 \leq i \leq n-1$: $u_iu_{i+1}$, $u_iv_i$, and $v_iv_{i+k}$, with subscripts modulo $n$~\cite{Watkins_gen}. The class of generalized Petersen graphs does not contain the graph in Figure~\ref{fig:tait13} but contains the graph of Figure~\ref{fig:cube}, since $Q=G(4,1)$. We consider the Nauru graph $N=G(12,5)$, which has several possible maximal matchings to be studied, and the famous Petersen graph $R=G(5,2)$, the only non Tait-colorable graph in the class of generalized Petersen graphs~\cite{Castagna1972}.

\begin{proposition} \label{prop:bipartiteubNauru}
Let $N$ be the Nauru graph. Then, $\eta(N) \leq 1/2$.
\end{proposition}

\begin{proof}
Let $M=\{e_1,\dots,e_8\}$ be a maximal matching of $N$ presented in Figure~\ref{fig:nauru}. The set $S$ of 8 vertices that are not end vertices of $e_1,\dots,e_8$ is an independent set of the graph. Hence, $M$ is a maximal matching that is not perfect. Moreover, each perfect matching of the graph has exactly 12 edges, and 8 of such edges must each saturate exactly one vertex of the independent set $S$. Therefore, each perfect matching of $N$ must have at most 4 edges in $M$. So, Lemma~\ref{maxmatchings} gives the upper bound.
\end{proof}

\begin{figure}[ht]
\centering
\includegraphics[width = .3\linewidth]{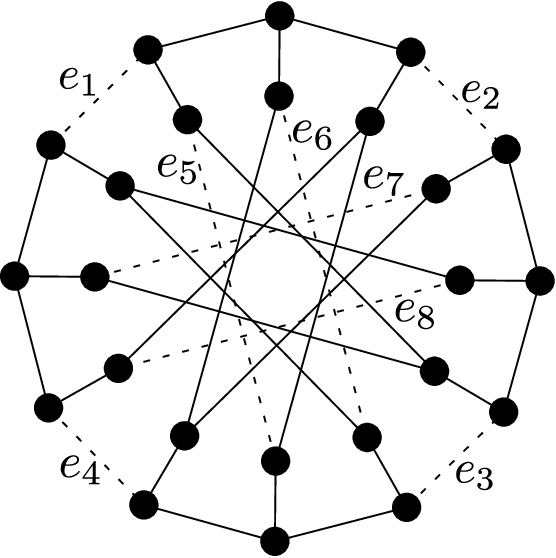}
\caption{Nauru graph $N$ with the maximal matching considered in the proof of Proposition~\ref{prop:bipartiteubNauru}.}
\label{fig:nauru}
\end{figure}

Similarly, we show that the Petersen graph $R = G(5,2)$ has $\eta(R) = 1/3$.

\begin{proposition} \label{lem:petersenub}
Let $R$ be the Petersen graph. Then, $\eta(R) = 1/3$.
\end{proposition}
\begin{proof}
By Lemma~\ref{lem:lb}, $\eta(R) \geq 1/3$. We now show that $\eta(R) \leq 1/3$. Let $M=\{e_1,e_2,e_3\}$ be the maximal matching presented in Figure~\ref{fig:petersen}. The set $S$ of 4 vertices that are not end vertices of $e_1,e_2,e_3$ is an independent set of the graph. Hence, $M$ is a maximal matching that is not perfect. Moreover, each perfect matching of the graph has exactly 5 edges, and 4 of such edges must each saturate exactly one vertex of the independent set $S$. Therefore, each perfect matching of $R$ must have at most 1 edge in $M$. So, Lemma~\ref{maxmatchings} gives again the upper bound.
\end{proof}

\begin{figure}[ht]
\centering
\includegraphics[width = .2\linewidth]{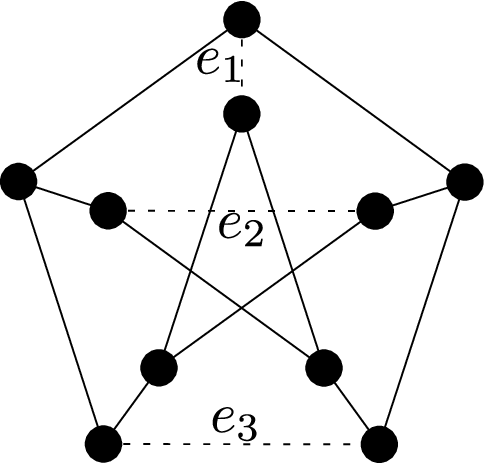}
\caption{Petersen graph $R$ with the maximal matching considered in the proof of Proposition~\ref{lem:petersenub}.} 
\label{fig:petersen}
\end{figure}

A snark is a nonplanar bridgeless cubic graph that is not Tait-colorable. A \emph{dot-product} of two cubic graphs $G$ and $H$ is any cubic graph obtained from $G\smallsetminus \{x,y\}$ and $H\smallsetminus \{aa', bb'\}$, where $x$ and $y$ are two adjacent vertices of $G$ and $aa'$ and $bb'$ is a matching of $H$. Let $a,a',b,b'$ be the four vertices of degree 2 in $H\smallsetminus \{aa', bb'\}$, and $x_1,x_2,y_1,y_2$ be the four vertices of degree 2 in $G\smallsetminus \{x,y\}$. We connect $H\smallsetminus \{aa', bb'\}$ to $G\smallsetminus \{x,y\}$ through a matching of size 4 in the resulting graph $ax_1,a'x_2,by_1,b'y_2$.

The dot-product is a famous operation for constructing infinitely many snarks, since the dot-product of two snarks is a snark. The two Blanu\v{s}a snarks $B^1$ and $B^2$ of order 18 were obtained by considering $G=H=R$, the Petersen graph~\cite{Blanusa}. Two infinite families Blanu\v{s}a First and Blanu\v{s}a Second (see Figure~\ref{fig:B1}) were subsequently defined by applying recursively the dot-product with $R$ starting respectively with $B^1$ and $B^2$~\cite{Watkins}.

\begin{proposition}\label{B1}
Let $B^1$ be the first member of the Blanu\v{s}a First family. Then, $\eta(B^1) \leq 2/5$.
\end{proposition}
\begin{proof}
Let $M=\{e_1,e_2,e_3,e_4,e_5\}$ be the matching of $B^1$ shown in Figure~\ref{fig:B1}. We claim that a perfect matching of $B^1$ can contain at most two edges of $M$. So, setting $w(e_1)=w(e_2)=w(e_3)=w(e_4)=w(e_5)=1$ and all other edge weights to $0$, we have the upper bound.

Indeed, to prove the claim, note first that the removal of $M$ from $B^1$
leaves four isolated vertices on the left, say $v_{2,5}, v_{3,5}, v_{2,4} $ and $v_{3,4}$,
and a matching of size two on the right.

If a matching $M'$ of $B^1$ contains $e_1$, and two other edges of the set
$\{e_2, e_3, e_4, e_5\}$,
then one of the vertices $v_{2,5}, v_{3,5}, v_{2,4} $ and $v_{3,4}$ cannot be saturated by $M'$.
If a perfect matching of $B^1$ contains three edges of the set $\{e_2, e_3,
e_4, e_5\}$, then it has to contain the fourth edge,
as otherwise the five remaining vertices on the right cannot be all
saturated. Now, the remaining six vertices on the left cannot be all
saturated.
\end{proof}

\begin{figure}[ht]
\centering
\includegraphics[scale = .3]{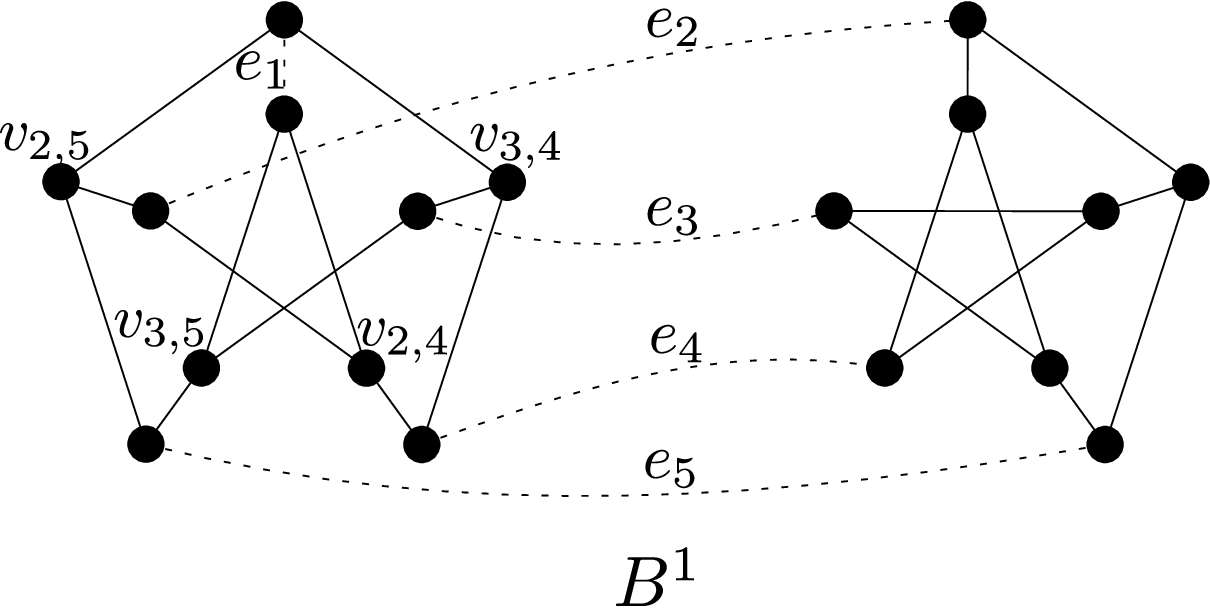} \hspace{.5cm} \includegraphics[scale=.3]{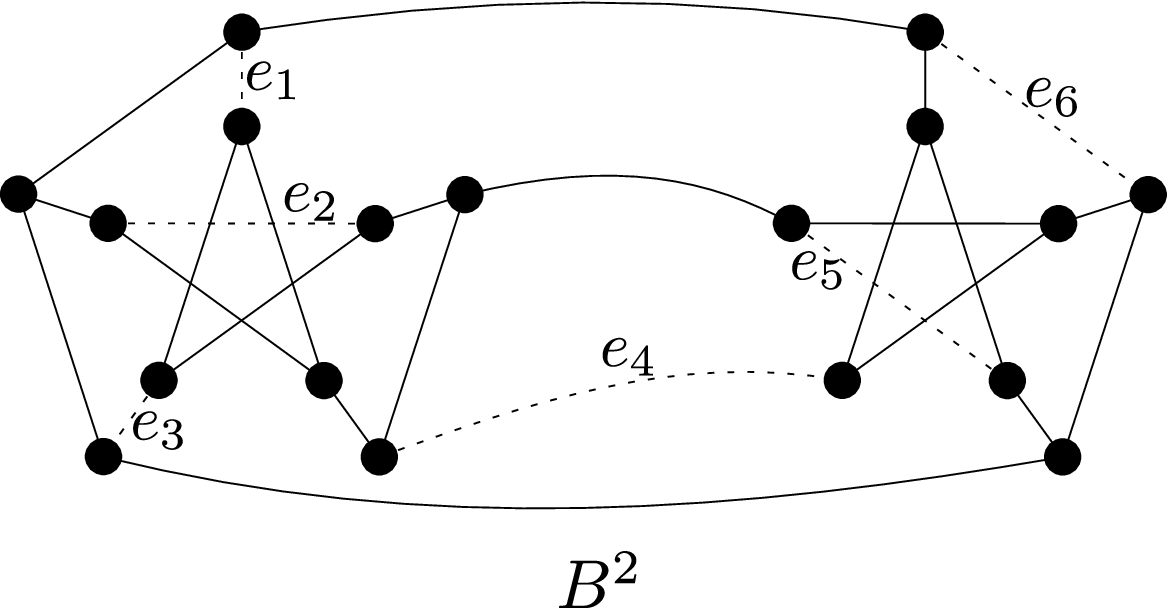}
\caption{Both Blanu\v{s}a snarks with the relevant matchings.}
\label{fig:B1}
\end{figure}

Note that a snark obtained by the dot-product of the Petersen graph (with two nonadjacent edges at distance 2 removed) and another snark (with two adjacent vertices removed) satisfies this upper bound, which applies to infinite families of snarks, including the historical Blanu\v{s}a First.

Unfortunately, we were not able to establish the same upper bound $2/5$ for the other Blanu\v{s}a snark, but Lemma~\ref{maxmatchings} provides the upper bound of $1/2$.

\begin{proposition}\label{B2}
Let $B^2$ be the first member of the Blanu\v{s}a Second family. Then, $\eta(B^2) \leq 1/2$.
\end{proposition}
\begin{proof}
Let $M=\{e_1,e_2,e_3,e_4,e_5,e_6\}$ be the maximal matching of $B^2$ shown in Figure~\ref{fig:B1}. The set $S$ of 6 vertices that are not end vertices of $e_1,e_2,e_3,e_4,e_5,e_6$ is an independent set of the graph. Hence, $M$ is a maximal matching that is not perfect. Moreover, each perfect matching of the graph has exactly 9 edges, and 6 of such edges must each saturate exactly one vertex of the independent set $S$. Therefore, each perfect matching of $B^2$ must have at most 3 edges in $M$. So Lemma~\ref{maxmatchings} gives the upper bound.
\end{proof}

The weaker bound obtained next applies to infinite families of snarks, including the historical Blanu\v{s}a Second.

\begin{proposition}\label{B2geral}
A snark obtained by the dot-product of the Petersen graph (with two nonadjacent edges at distance 1 removed) and another snark (with two adjacent vertices removed) satisfies the upper bound of $2/3$.
\end{proposition}
\begin{proof}
Let $M=\{e_1,e_2,e_3,e_4,e_5,e_6\}$ be the matching of $B^2$ shown in Figure~\ref{fig:B2}. We claim that a perfect matching of $B^2$ can contain at most four edges of $M$. So, setting $w(e_1)=w(e_2)=w(e_3)=w(e_4)=w(e_5)=w(e_6)=1$ and all other edge weights to $0$, we have the upper bound.

Indeed, to prove the claim, note first that the removal of $M$ from $B^2$
leaves two isolated vertices on the left, say $a$ and $b$,
and a matching of size two on the right.

If a matching $M'$ of $B^2$ contains 5 edges of $M$ and is such that $M'$ contains $\{e_3,e_4,e_5,e_6\}$ and exactly one edge of $\{e_1, e_2\}$, 
then one of the vertices $a, b$ cannot be saturated by $M'$.
If a perfect matching of $B^2$ contains 3 edges of the set $\{e_3,e_4,e_5,e_6\}$ and also $e_1$ and $e_2$, then it has to contain the fourth edge of the set $\{e_3,e_4,e_5,e_6\}$,
as otherwise the five remaining vertices on the right cannot be all
saturated. Now, the remaining two vertices $a$ and $b$ on the left cannot be all
saturated.
\end{proof}

\begin{figure}[ht]
\centering
\includegraphics[scale = .3]{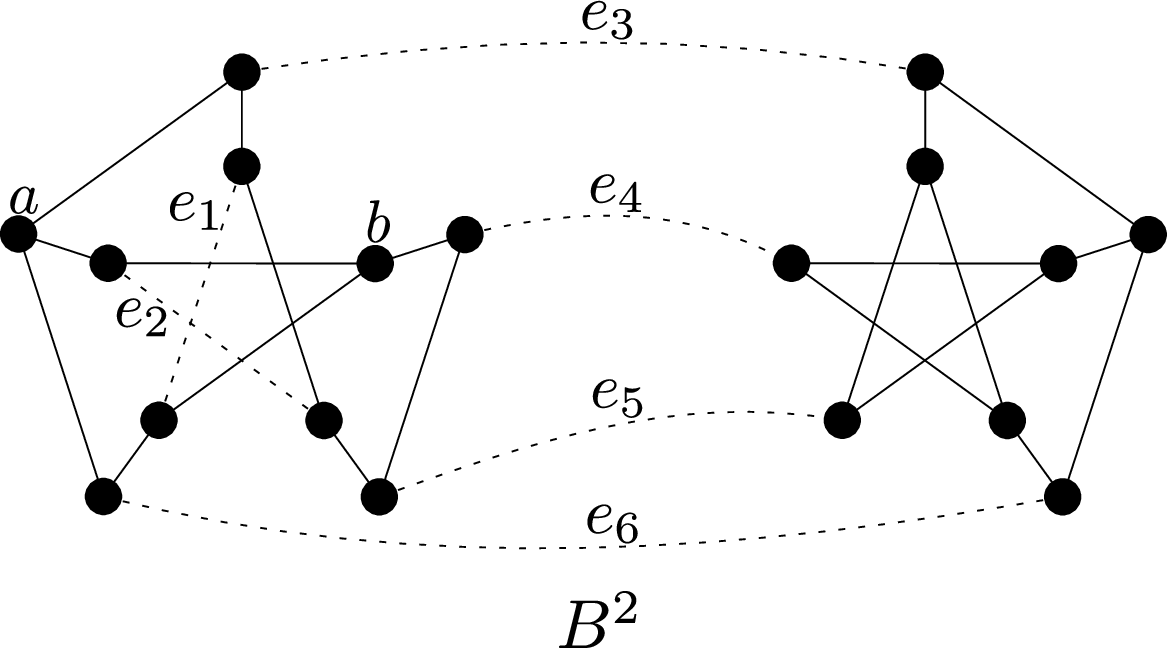}
\caption{The Blanu\v{s}a snark $B^2$ with the matching considered in Proposition~\ref{B2geral}.}
\label{fig:B2}
\end{figure}

\paragraph{A snark family obtained by the 2-construction}
Starting from the Petersen graph $R$ with its matching $M$ from Figure~\ref{fig:petersen}, it is possible to obtain an infinite family of snarks with the same upper bound of $1/3$ using a classical construction in the theory of cubic graphs~\cite{isaacs75,Sasaki}, as follows. To construct a new member $G$ of the family with a matching $M_G$, start from two already constructed graphs $H_1$ and $H_2$ with their matchings $M_{H_1}$ and $M_{H_2}$, and remove from them edges $uv \not\in M_{H_1}$ and $xy \not\in M_{H_2}$. Without loss of generality, we can assume that $u$ is $M_{H_1}$-saturated and $v$ is not, and likewise that $x$ is $M_{H_2}$-saturated and $y$ is not. Now we let $G$ be the graph obtained from these graphs by joining $u$ with $y$ and $v$ with $x$. Finally, let $M_G$ be the union of $M_{H_1}$ and $M_{H_2}$, and note that in $G$ it is also true that any edge not in $M_G$ has one end that is $M_G$-saturated and one that is not. Any graph $G$ with a matching $M_G$ obtained by this construction will have $|M_G| = 3|V(G)|/10$, and from Lemma~\ref{maxmatchings} we have $\eta(G) \leq 1/3$.

All previous bridgeless cubic graphs for which we showed an upper bound of $1/3$ are either hamiltonian or nonbipartite. Next, we show that $\eta(\NN) = 1/3$, where $\NN$ is the class of nonhamiltonian bipartite cubic graphs.

\begin{lem} \label{lem:bipartiteub}
Let $\NN$ be the class of nonhamiltonian bipartite cubic graphs. Then $\eta(\NN) = 1/3$.
\end{lem}
\begin{proof}
First, we prove that the cubic graph $G$ of Figure~\ref{fig:spbn} is nonhamiltonian. $G$ is obtained by the following construction. Let $G_1$ and $G_2$ be the two cubic graphs shown in Figure \ref{fig:spbn}. We remove vertex $v_1$ from $G_1$ and vertex $v_2$ from $G_2$. Now, let $z_0,z_1,z_2$ be the neighbors of $v_1$ in $G_1$ and $u_0,u_1,u_2$ be the neighbors of $v_2$ in $G_2$. We add all edges $z_0u_0,z_1u_1,z_2u_2$, which gives us graph $G$ which is clearly cubic and bipartite. By contradiction, suppose that $G$ is hamiltonian and let $C$ be a hamiltonian cycle. So exaclty two of the three edges $z_0u_0,z_1u_1,z_2u_2$ are in $C$. Without loss of generality, we may assume that $z_0u_0,z_1u_1$ belong to $C$. Now we obtain a hamiltonian cycle in $G_1$ by replacing the edges of $C$ with at least one endvertex in $\{u_0,u_1,u_2\}$ by the edges $z_0v_1,v_1z_1$. But this contradicts a result of \cite{asano82}, where it was shown that $G_1$ is nonhamiltonian.

\begin{figure}[ht!]
\centering
\includegraphics[scale = .43]{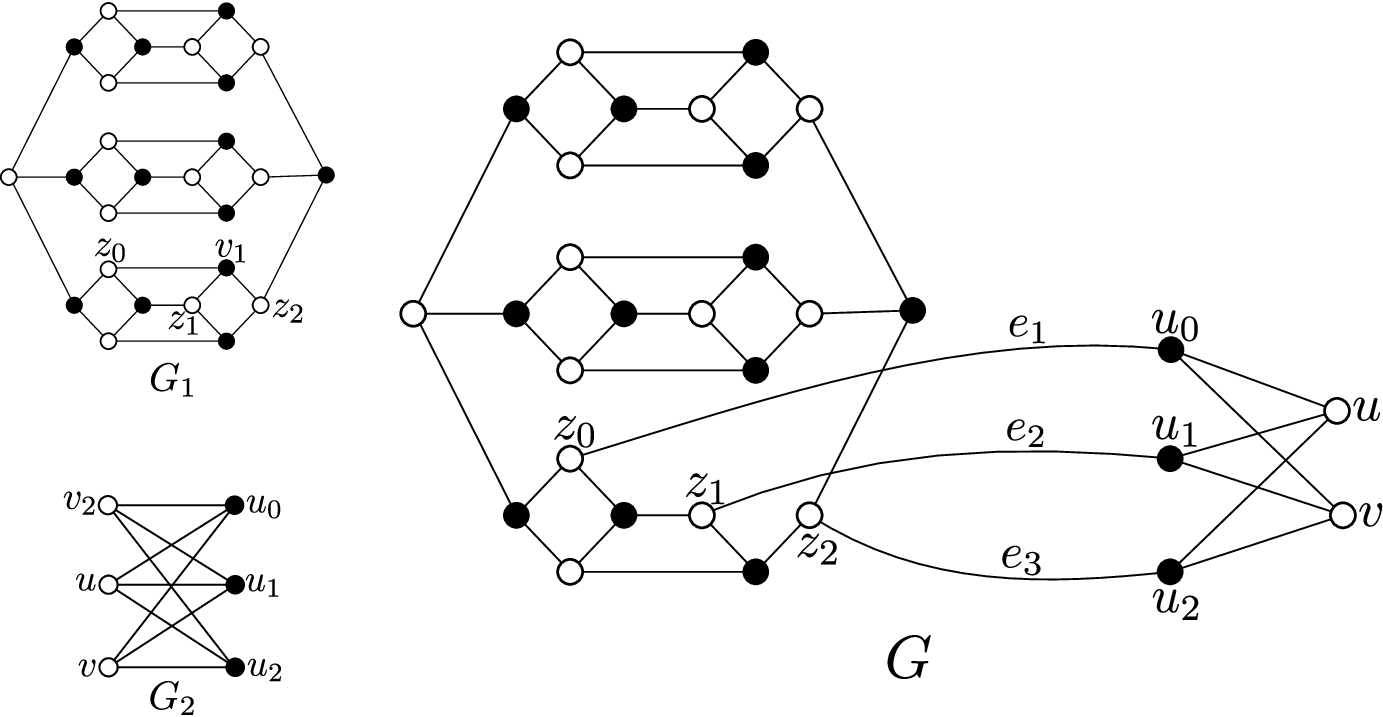}
\caption{A nonhamiltonian bipartite cubic graph $G$ obtained from $G_1$ and $G_2$ and containing no perfect matching with more than one edge among $e_1,e_2,e_3$.}
\label{fig:spbn}
\end{figure}

By Lemma~\ref{lem:lb}, $\eta(\NN) \geq 1/3$. We will show that $\eta(\NN) \leq 1/3$. Consider the nonhamiltonian bipartite cubic graph $G=(V,E)$ represented in Figure~\ref{fig:spbn} and let $e_1,e_2,e_3$ be the edges as shown in the figure.
Consider an edge weighting $w:E\rightarrow \RE^+$ such that $w(e_1)=w(e_2)=w(e_3)=1$ and $w(e)=0$, $\forall e\in E\setminus\{e_1,e_2,e_3\}$. Notice that there exists no perfect matching that contains more than one edge of $e_1,e_2,e_3$. Indeed, if a matching contains both $e_1$ and $e_2$, then the white vertices $u$ and $v$ cannot be both saturated. But clearly, there exists a perfect matching containing exactly one of the edges $e_1,e_2,e_3$. Thus, $w(P^*(G))=1$ and $w(M^*(G))=3$, and hence $\eta(G)\leq 1/3$ which implies that $\eta(\NN)\leq 1/3$.
\end{proof}

The following theorem combines the upper and lower bounds on $\eta$ for several classes of bridgeless cubic graphs, and follows immediately from the previous results. Figure~\ref{fig:diagram} illustrates the results on cubic graphs presented in this work.

\begin{thm}
The following graph classes have $\eta = 1/3$: Tait-colorable graphs, planar bridgeless cubic graphs, hamiltonian cubic graphs, generalized Petersen graphs, all members of an infinite family of snarks obtained by the 2-construction, and nonhamiltonian bipartite cubic graphs.
\end{thm}

\begin{figure}[ht!]
\centering
\includegraphics[scale = .5]{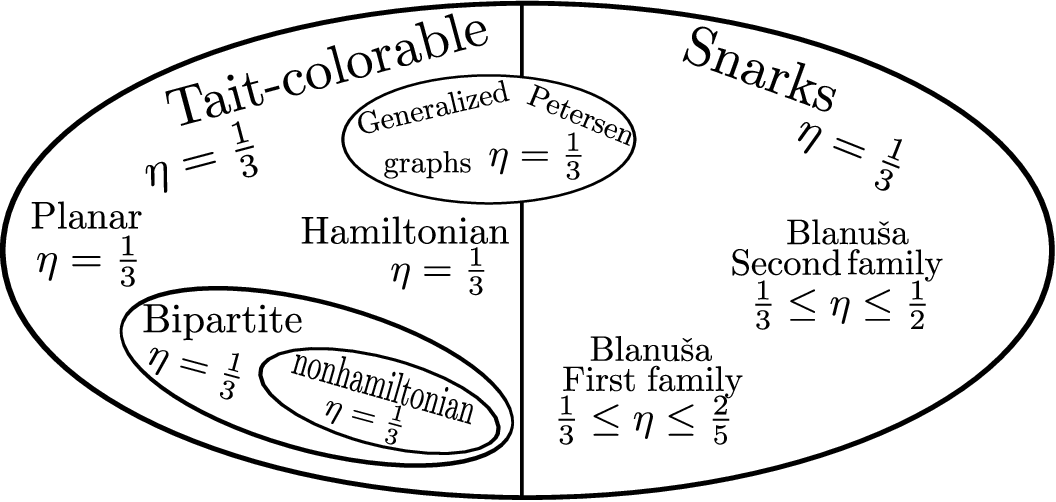}
\caption{A depiction of the results on bridgeless cubic graphs presented in this work.}
\label{fig:diagram}
\end{figure}

\section{Regular Bipartite Graphs} \label{s:regular}

A $\Delta$-regular bipartite graph is a bipartite graph such that every vertex has degree $\Delta$. In this section, we show that for all $\Delta \geq 2$, the class $\BB_\Delta$ of $\Delta$-regular bipartite graphs satisfies $\eta(\BB_\Delta) = 1/\Delta$.

\begin{thm} \label{lem:lbip}
Let $\BB_\Delta$ be the class of $\Delta$-regular bipartite graphs. Then, $\eta(\BB_\Delta) =1/\Delta$, for all $\Delta\geq 2$.
\end{thm}
\begin{proof}
We start by showing that $\eta(\BB_\Delta)\geq 1/\Delta$. It follows from K\"{o}nig's Theorem~\cite{konig}, that all $\Delta$-regular bipartite graphs are $\Delta$-edge-colorable, i.e, they admit $\Delta$ disjoint perfect matchings such that each edge is contained in exactly one perfect matching. Let $G=(V,E)$ be a $\Delta$-regular bipartite graph and let $P_1,\ldots,P_{\Delta}$ denote these $\Delta$ disjoint perfect matchings. Let $w:E\rightarrow \RE^+$ be a weighting of $E$. We may assume, without loss of generality, that $w(P_1) \geq \cdots \geq w(P_{\Delta})$. Let $M^*(G)$ be a maximum weight matching in $G$. Then, clearly
\[w(M^*(G)) \leq w(P_1) + \cdots + w(P_{\Delta}) \leq \Delta w(P_1) \leq \Delta w(P^{*}(G))\]
where $P^*(G)$ is a maximum weight perfect matching and therefore $\eta(G) \geq 1/\Delta$. Since $G$ is an arbitrary $\Delta$-regular bipartite graph, it follows that $\eta(\BB_\Delta) \geq 1/\Delta$.

Now we prove that for each $\Delta\geq 2$, there exists a $\Delta$-regular bipartite graph $G$ such that $\eta(G) \leq 1/\Delta$. Consider the $\Delta$-regular bipartite graph $G=(V,E)$ represented in Figure~\ref{fig:Reg_bip}. Let $e_1,e_2,\cdots,e_{\Delta}$ be the edges as shown in Figure~\ref{fig:Reg_bip}.
Consider a weighting $w:E\rightarrow \RE^+$ of $E$ such that $w(e_1)=\cdots=w(e_{\Delta})=1$ and $w(e)=0$, $\forall e\in E\setminus \{e_1,\cdots,e_{\Delta}\}$. We claim that there exists no perfect matching that contains more than one edge among $e_1,e_2,\cdots,e_{\Delta}$. Indeed, suppose that $P$ is a perfect matching of $G$ containing both $e_1,e_2$. In order to saturate all black vertices on the left, we need $\Delta-1$ white vertices on the left. But since $P$ contains both $e_1$ and $e_2$, only $\Delta-2$ such vertices are available. Hence $P$ cannot exist. Furthermore, there clearly exists a perfect matching containing exactly one labeled edge. Thus $w(M^*(G))=\Delta$ and $w(P^*(G))=1$ and hence $\eta(G)\leq 1/\Delta$. We conclude that $\eta(\BB_\Delta)\leq 1/\Delta$.
\end{proof}

\begin{figure}[ht]
\centering
\includegraphics[scale = .5]{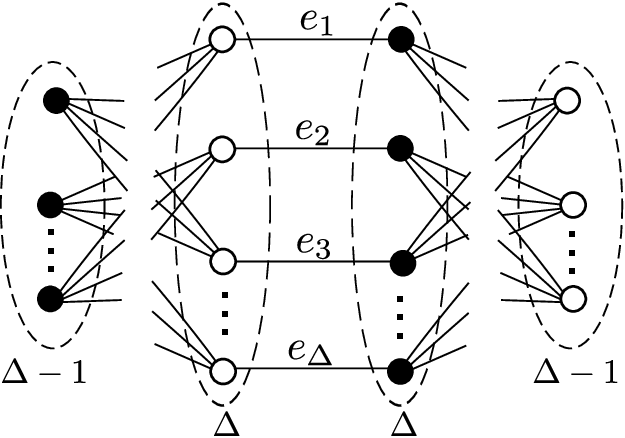}
\caption{A $\Delta$-regular bipartite graph with no perfect matching containing more than one edge among $e_1,e_2,\cdots,e_{\Delta}$.}
\label{fig:Reg_bip}
\end{figure}

\section{Conclusion}

We introduced the parameter $\eta$ to quantify the cost of perfection for matchings. The parameter was motivated by an application in computer graphics, but we believe that it is a natural and challenging parameter from a theoretical viewpoint. We characterized graphs with extreme values of $\eta$, we provided tight bounds for $\eta$ for several classes of cubic graphs, and determined the exact value of $\eta$ for the class of regular bipartite graphs.

The dual graph of a triangulation is a bridgeless cubic graph and many recent works have been devoted to quadrangulations~\cite{gopi04,remacle11,daniels11,tri2quad,lizier10a}.
The specific classes studied here are related to important triangulations for computer-graphic applications. 
For instance, hamiltonian  meshes are used to accelerate the graphics pipeline~\cite{arkin96,eppstein04}, also bipartite cubic graphs (planar or not) can be used to improve the rendering of 3D geometric models~\cite{sander08}. 
The obtained bounds aid the decision process of whether to use a perfect matching or, alternatively, use a two step quadrangulation method~\cite{tarini10}, which first obtains a maximum weight matching and then deals with the unmatched triangles.

Many open problems still remain.
We propose to extend the construction which gives $\eta = 1/3$ to the Petersen graph to other infinite families of snarks.
Another possible direction of work consists of calculating $\eta$ for cubic graphs that are dual of 4-8~meshes~\cite{velho01,velho03,velho04} (Figure~\ref{fig:48mesh}). Such meshes received a lot of attention recently~\cite{amorim12,goes08,weber07,goldenstein05}.
Furthermore, the problem of bounding $\eta$ is interesting \emph{per se}, therefore it is natural to investigate the value of $\eta$ for other graph classes whose graphs admit a perfect matching, such as grid graphs.

Finally, as suggested by a referee, an intriguing question is whether there is a polynomial algorithm to determine the value of $\eta(G)$ for a given graph $G$. In particular, can we determine in polynomial time if a given bridgeless cubic graph $G$ has $\eta(G)=1/3$?

\begin{figure}[ht]
\centering
\includegraphics[width = .5\linewidth]{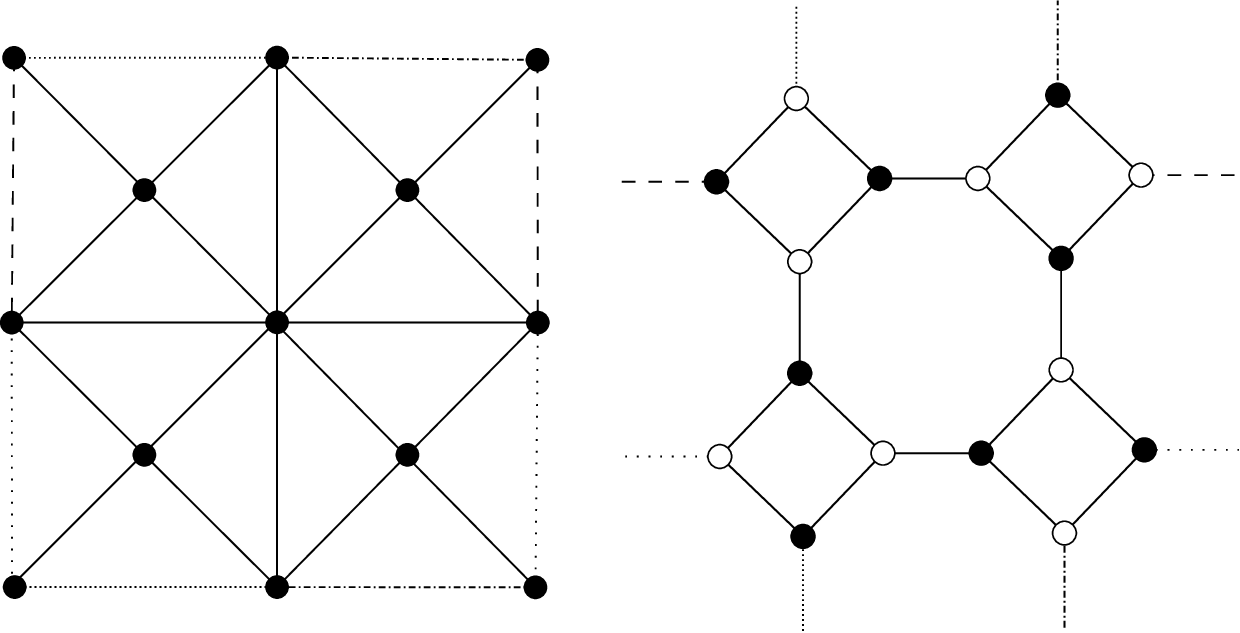}
\caption{A 4-8~mesh (left) and its dual bipartite cubic graph (right).}
\label{fig:48mesh}
\end{figure}

\section*{Acknowledgments}

The authors would like to thank Bernard Ries, Hugo Nobrega, and Fran\c{c}ois Dross for the insightful discussions,
Vahan Mkrtchyan for the proof of Lemma~\ref{lem:lb}, and Stanford Computer
Graphics Laboratory for the bunny model.

\bibliographystyle{plain}
\bibliography{graph}

\end{document}